\documentclass{article}
\usepackage{macros}

\newif\ifdraft
\draftfalse

\newif\ifnames
\namestrue

\newcommand{\sasha}[1]{\ifdraft{\color{olive} \;#1 --Sasha\;}\fi}

\newcommand{\ted}[1]{\ifdraft{\color{red} \;#1 --Ted\;}\fi}
\newcommand{\ian}[1]{\ifdraft{\color{blue} \;#1 --Ian\;}\fi}
\newcommand{\unfinished}[1]
{\ifdraft{\color{orange} #1\;}\fi}

\newcommand*\samethanks[1][\value{footnote}]{\footnotemark[#1]}

\title{Collapsing Catalytic Classes}
\author{\ifnames 
        Michal Kouck\'{y}\thanks{Partially supported by the Grant Agency of the Czech Republic under the grant agreement no. 24-10306S and by Charles Univ. project UNCE 24/SCI/008.} \\ Charles University \\ \texttt{koucky@iuuk.mff.cuni.cz} \\ \\
        Edward Pyne\thanks{Supported by the NSF GRFP.} \\ MIT \\ \texttt{epyne@mit.edu}
        \and
        Ian Mertz\samethanks[1] \\ Charles University \\ \texttt{iwmertz@iuuk.mff.cuni.cz}  \\ \\
        Sasha Sami\samethanks[1] \\ Charles University \\ \texttt{sashasami@iuuk.mff.cuni.cz}
        \fi}
\date{\today}

\begin{document}
\begin{titlepage}
\maketitle

\begin{abstract}
    A catalytic machine is a space-bounded Turing machine with additional
    access to a second, much larger work tape, with the caveat that this
    tape is full, and its contents must be preserved by the computation.
    Catalytic machines were defined by Buhrman et al. (STOC 2014), who,
    alongside many follow-up works, exhibited the power of
    catalytic space ($\CSPACE$) and in particular catalytic logspace
    machines ($\CL$) beyond that of traditional space-bounded machines.

    Several variants of $\CL$ have been proposed, including non-deterministic and co-non-deterministic catalytic computation by Buhrman et al. (STACS 2016)
    and randomized catalytic computation by Datta et. al. (CSR 2020). These and other works proposed several questions, such as catalytic analogues of the theorems of Savitch and Immerman and Szelepcs\'{e}nyi. Catalytic computation was recently derandomized by Cook et al. (STOC 2025), but only in certain parameter regimes.

    We settle almost all questions regarding non-deterministic and randomized
    catalytic computation, by giving an optimal reduction from catalytic space with additional resources to the corresponding non-catalytic space classes. With regards to non-determinism, our main result is that
    \[
    \CL = \CNL
    \]
    and with regards to randomness, we show
    \[
    \CL = \CPL
    \]
    where $\CPL$ denotes randomized catalytic logspace where the accepting probability can be arbitrarily close to $1/2$.
    We also have a number of near-optimal partial results for non-deterministic and randomized catalytic computation with less catalytic space. In particular, we show catalytic versions of Savitch's theorem, Immerman-Szelepsc\'{e}nyi, and the derandomization results of Nisan and Saks and Zhou, all of which are unconditional and hold for all parameter settings.
    
    Our results build on the \textit{compress-or-compute framework} of
    Cook et al. (STOC 2025).
    Despite proving broader and stronger results, our framework is
    simpler and more modular.
\end{abstract}

\thispagestyle{empty}
\newpage
\unfinished{
WRITING:

\begin{tabular}{c|c}
   Section  &  Final status good?\\
   Abs  & \sasha{Read} \ian{Read} \\
   Intro & \sasha{Read} \ian{Read} \\
   Prelims & \sasha{Read} \ian{Read} \\ 
   Overview & \sasha{Read} \ian{Read} \\
   Formal Proof & \sasha{Read}
\end{tabular}

FOR FUTURE
\begin{itemize}
    \item Can we get $\NL\subseteq \CUSPACE[O(s),O(s^2)]$? I.e. improve Pyn24 to get deterministic decompression.
\end{itemize}

TODOs
\begin{itemize}
    \item Ian: Discussion section in intro (all hierarchies and UL stuff go here) \ted{Looked over} \ian{More reorganization of the paper but no substantive changes.}
    \item Michal: proof nitpick fixes
    \item Michal: pass on abs/intro/prelims
    \item why are the authors so weirdly misaligned?
\end{itemize}
}

\end{titlepage}


\section{Introduction}
\label{sec:intro}

In this paper we study space-bounded classes with access to additional
resources. The three resources we study are \emphdef{non-determinism},
\emphdef{randomness}, and, most central to this work, \emphdef{catalytic space}.

The first two models, i.e. non-deterministic and randomized space,
have long yet unresolved histories going back to the earliest days
of theoretical computer science.
The first major result, due to Savitch in
1970~\cite{Savitch70}, states that determinism can simulate
non-determinism with only a quadratic space overhead, i.e.
$\NSPACE[s] \subseteq \SPACE[s^2]$.
Much later, Nisan~\cite{Nisan92} and Saks and Zhou~\cite{SaksZhou99}
proved that bounded-error randomness can similarly be simulated with
an even smaller blowup, showing $\BPSPACE[s] \subseteq \SPACE[s^{3/2}]$.
In between, Borodin, Cook, and Pippenger~\cite{BorodinCookPippenger83}
derandomized \textit{unbounded-error} randomized computation with a quadratic space blowup, i.e. $\PrSPACE[s] \subseteq \SPACE[s^2]$.

While derandomization in the bounded-error case continues to see vigorous
work---the exponent
has since been further improved by Hoza~\cite{Hoza21} by a $o(1)$
factor---Savitch's Theorem remains the best known simulation for
non-determinism to date. Subsequent results by Immerman and
Szelepsc\'{e}nyi~\cite{Immerman88,Szelepcs_nyi_1988}
and Reinhardt and Allender~\cite{ReinhardtAllender00} show progress
from a different angle, namely by showing that $\NL$ is closed under
complement ($\NL = \coNL$), and $\NL$
can be made unambiguous ($\NL = \UL$) assuming strong circuit lower bounds,
respectively. Meanwhile, while \cite{BorodinCookPippenger83} shows
that $\PL=\PrSPACE[O(\log n)] \subseteq \TCo$, a slightly stronger containment than $\Logspace^2$,
there have been no other improvements
on unbounded-error derandomization.

\subsection{Catalytic Computation}
We now turn our attention to our third resource.
The \emphdef{catalytic space} model, introduced by Buhrman, Cleve, Kouck{\'{y}}, Loff,
and Speelman~\cite{BuhrmanCleveKouckyLoffSpeelman14},
studies the question of whether full space can be useful to computation.
In $\CSPACE[s,c]$ we consider a typical $\SPACE[s]$ machine augmented
with a second work tape, called the \emphdef{catalytic tape}, which has length $c$.
We think of $c$ to be much larger than $s$, often exponentially larger;
however, this tape is initialized to some arbitrary string $\tau$, and at the end
of the computation our machine must reset the catalytic tape to the starting $\tau$.

Despite this restriction, \cite{BuhrmanCleveKouckyLoffSpeelman14} show that such machines
are unexpectedly powerful. Focusing on the class called \emphdef{catalytic logspace}
($\CL := \CSPACE[O(\log n),\poly(n)]$), they show that the catalytic tape
is at least as powerful as randomization and non-determinism
($\BPL$ and $\NL$, respectively), and contains problems (e.g. determinant)
which are thought to be in neither.

Catalytic computation appeared in the context of \emphdef{composition} for
space-bounded functions, where it was unknown whether computing multiple instances
of a function causes the space complexity to scale up linearly in tandem.
Such techniques and insights were crucially used in the recent result of
Cook and Mertz~\cite{CookMertz21,CookMertz22,CookMertz24} on the tree evaluation problem,
which was later used in a breakthrough by Williams~\cite{Williams25} showing
$\TIME[t]$ can be simulated in only $\sqrt{t \log t}$ space.

\subsection{Non-Deterministic and Randomized Catalytic Computation}

In light of the surprising power of catalytic space,
follow-up works have proposed several variants of the base model,
such as non-deterministic catalytic
computation~\cite{BuhrmanKouckyLoffSpeelman18,GuptaJainSharmaTewari19,Sdroievski24},
randomized catalytic computation~\cite{DattaGuptaJainSharmaTewari20,CookLiMertzPyne24},
non-uniform catalytic computation~\cite{GirardKouckyMckenzie15,Potechin17,RobereZuiddam21,CookMertz22}, and
error-prone catalytic computation~\cite{GuptaJainSharmaTewari24,FolkertsmaMertzSpeelmanTupker25},
to name a few
(see surveys of Kouck\'{y}~\cite{Koucky16} and Mertz~\cite{Mertz23} for an overview).

Non-deterministic catalytic space was introduced by
Buhrman et al.~\cite{BuhrmanKouckyLoffSpeelman18}, who showed that,
assuming pseudorandom generators, \emphdef{catalytic non-deterministic
logspace} ($\CNL$) is also closed under complement ($\CNL=\coCNL$).
Later work of~\cite{GuptaJainSharmaTewari19} extended this by showing,
again assuming pseudorandom generators, that $\CNL$ can also be made
unambiguous ($\CNL = \CUL$).
The question of proving such statements unambiguously, or of obtaining
a catalytic analogue of Savitch's Theorem,
was put forth several times as an open question~\cite{BuhrmanKouckyLoffSpeelman18,GuptaJainSharmaTewari19,Koucky16,Mertz23,CookLiMertzPyne24}.

Randomized catalytic space was introduced by Datta et al.~\cite{DattaGuptaJainSharmaTewari20}, where they showed
that \emphdef{catalytic randomized logspace}
($\CBPL$) equals $\CL$ under similar pseudorandomness assumptions.
This was recently shown unconditionally by
Cook et al.~\cite{CookLiMertzPyne24}; however, their result only holds for
$\CSPACE[s,c]$ when $c = 2^{\Theta(s)}$,
while the general case is of interest in other settings
\cite{BisoyiDineshSarma22,Pyne24,FolkertsmaMertzSpeelmanTupker25}.
Thus far there has been no study of \textit{unbounded} error randomness in
catalytic computation ($\CPL$).\footnote{Terminology for unbounded error
randomized classes varies between using Pr and simply P, but as we will
be introducing e.g. $\CPSPACE$ the latter creates too many collisions
between ``probabilistic'' and ``polynomial''.}

\subsection{Our Results}
We settle essentially every question regarding non-deterministic and randomized
catalytic computation.
Our most striking result is that, with access to a large pre-filled hard drive,
neither non-determinism nor unbounded-error randomness gives any additional power:
\begin{theorem}~\label{cor:collapse-cnl}
    $$\CL = \CNL.$$
\end{theorem}
\begin{theorem}~\label{cor:collapse-cpl}
    $$\CL = \CPL.$$
\end{theorem}
\noindent
From these results we conclude in the context of space-bounded computation,
catalytic space acts as one resource to rule them all.
We note that before our result there were no non-trivial connections between deterministic
catalytic computation and either non-deterministic or unbounded-error randomized
catalytic computation, even under assumptions.

We also have a number of results for other values of $s$ and $c$. 
First, a catalytic Savitch's Theorem holds, with overhead matching that of the non-catalytic case:
\begin{theorem}\label{cor:cat-savitch}
    For all $s := s(n)$, $c := c(n)$ such that $\log n \leq s \leq c \leq 2^s$,
    $$\CNSPACE[s,c] \subseteq \CSPACE[O(s^2),O(c)].$$
\end{theorem}
\noindent
No results of the above form were known before, even under assumptions.

Second, $\CNSPACE[s,c]$ is closed under complement:
\begin{theorem}\label{cor:cat-IS}
    For all $s := s(n)$, $c := c(n)$ such that $\log n \leq s \leq c \leq 2^s$,
    $$\coCNSPACE[s,c] \subseteq \CNSPACE[O(s),O(c)].$$
\end{theorem}
\noindent
This result was previously known to follow from strong lower bounds~\cite{BuhrmanKouckyLoffSpeelman18}. 

Third, $\CBPSPACE[s,c]$ collapses to $\CSPACE[O(s),O(c)]$ for every $c\ge s^2$,
as well as tradeoffs for all other $c$:
\begin{theorem}\label{cor:cat-derand}
    For all $s := s(n)$, $c := c(n)$ such that $\log n \leq s \leq c \leq 2^s$
    and all $\alpha \in [0,1/2]$,
    $$\CBPSPACE[s,c] \subseteq \CSPACE[O(s^{1+\alpha}),O(c+s^{2-\alpha})].$$
\end{theorem}
\noindent
This strongly extends the result of \cite{CookLiMertzPyne24} as well as
that of Pyne~\cite{Pyne24}, who showed the same result in the restricted case of $\CSPACE[s,c=0]$ on the left hand side.

Finally, $\CPSPACE[s,c]$ also obeys a Savitch-like theorem:
\begin{theorem}\label{cor:cat-derand-pl}
    For all $s := s(n)$, $c := c(n)$ such that $\log n \leq s \leq c \leq 2^s$,
    $$\CPSPACE[s,c] \subseteq \CSPACE[O(s^2),O(c)].$$
\end{theorem}

\section{Catalytic Machines}
We first define catalytic Turing machines:
\begin{definition}[Catalytic space]
    A \emphdef{catalytic Turing Machine} with \emphdef{free space $s := s(n)$}
    and \emphdef{catalytic space $c := c(n)$} is a Turing machine with the following tapes:
    \begin{enumerate}
        \item a read-only input tape of length $n$ which is initialized to $x \in \{0,1\}^n$
        \item a read-write work tape of length $s$ which is initialized to $0^{s}$
        \item a read-write \emphdef{catalytic tape} of length $c \leq 2^s$ which is
        initialized to some $\tau \in \{0,1\}^{c}$
    \end{enumerate}
\end{definition}

Such machines can be augmented with additional resources in the traditional
manner, but as with ordinary space-bounded computation,
these machines need to be defined with care.
\ian{Following text seems like it needs a lot of reworking.}
At each time step, a non-deterministic/randomized machine has two
(not necessarily distinct) choices for its transition: the 0-choice and 1-choice.
Sometimes we can think of the choices as being selected according to the content of
an auxiliary non-deterministic/random tape that provides read-only one-way access to its
content. This view will be useful later to define the space bounded hierarchy.
For machines that compute binary functions we think of the machine as outputting 1 if it reaches an accepting state, and outputting 0 if it reaches a rejecting state.
We equip machines that output a value from a larger range with an output tape that provides write-only one-way access.
The machine is expected to write its output on this output tape.
Oracle machines are equipped with a write-only one-way access oracle tape. 
To make an oracle query, the machine writes its query on the oracle tape, issues a query request to its oracle, the tape is reset to empty and the machine transitions into its next state depending on the query answer.
A catalytic Turing machine $\cM$ is said to be \emphdef{valid} if for every $x$ and $\tau$, the machine halts in finite time with the catalytic tape containing $\tau$ regardless of its non-deterministic/random choices.
(In particular, the machine is not allowed to loop forever.)

\begin{definition}[Variants of catalytic space] \label{def:cspace-variants}
    For a boolean function $f$, we say a catalytic machine computes $f$ if
    \begin{itemize}
        \item \emphdef{(deterministic machine)} for every $x$, the machine outputs $f(x)$.
        \item \emphdef{(non-deterministic machine)} for every $x$, if $f(x)=1$ then there is a sequence on non-deterministic choices where the machine accepts $x$, and if $f(x)=0$ then for any sequence on non-deterministic choices the machine rejects.  A \emphdef{co-non-deterministic} machine must always accept when $f(x)=1$ and sometimes reject when $f(x)=0$.
        \item \emphdef{(bounded-error randomized machine)} for every $x$, the machine outputs $f(x)$ with probability at least 2/3 over its random choices. 
        \item \emphdef{(unbounded-error randomized machine)} for every $x$, the machine outputs $f(x)$ with probability greater than 1/2 over its random choices. 
    \end{itemize}
\end{definition}

\begin{definition}[Complexity classes]
    A catalytic Turing machine decides a language $L$ if it computes the characteristic function of $L$. For $s := s(n)$ and $c := c(n)$ we define
    \begin{itemize}
        \item $\CSPACE[s,c]$ to be the class of languages which can be computed by
        a catalytic Turing machine with free space $s$ and catalytic
        space $c$.
        \item $\CNSPACE[s,c]$ ($\coCNSPACE[s,c]$) to be the class of languages which can be
        computed (co-computed) by a non-deterministic catalytic Turing machine with
        free space $s$ and catalytic space $c$.
        \item $\CBPSPACE[s,c]$ to be the class of languages which can be
        computed by a bounded-error randomized catalytic Turing machine with
        free space $s$ and catalytic space $c$.
        \item $\CPSPACE[s,c]$ to be the class of languages which can be
        computed by an unbounded-error randomized catalytic Turing machine with
        free space $s$ and catalytic space $c$.
    \end{itemize}
\end{definition}


As discussed above, our main focus is on class catalytic logspace,
where we fix the parameters $s$ and $c$ to be logarithmic and polynomial in $n$,
respectively.
\begin{definition}[Catalytic logspace classes]
    We define the following instantiations:
    \begin{itemize}
        \item $\CL := \bigcup_{d \in \mathbb{N}} \CSPACE[d \log n,n^{d}]$
        \item $\CNL := \bigcup_{d \in \mathbb{N}} \CNSPACE[d \log n,n^{d}]$
        \item $\CBPL := \bigcup_{d \in \mathbb{N}} \CBPSPACE[d \log n,n^{d}]$
        \item $\CPL := \bigcup_{d \in \mathbb{N}} \CPSPACE[d \log n,n^{d}]$
    \end{itemize}
\end{definition}


\section{Main Technical Theorem}

All of our results follow from a generic reduction from catalytic space with additional resource $\cB$, to the corresponding non-catalytic space class: 
\begin{theorem}\label{thm:main-tech}
    Let $\mathcal{B} \in \{\text{\upshape\textsf{N}},\text{\upshape\textsf{coN}},\text{\upshape\textsf{BP}},\text{\upshape\textsf{Pr}}\}$.
    Then for all $s := s(n)$, $c := c(n)$ such that $\log n \leq s \leq c \leq 2^s$,
    $$\CBSPACE[s,c] \subseteq  \CSPACE[O(s),O(c)]^{\BSPACE[O(s)]}$$
    where queries made to the oracle are on inputs of length $2^{O(s)}$, and for $\mathcal{B} \in \{\text{\upshape\textsf{BP}},\text{\upshape\textsf{Pr}}\}$ the oracle
    is for the corresponding promise class.
\end{theorem}

While \Cref{thm:main-tech} is stated in terms of oracles, this is just a
convenience for the terseness of our statement and proof; as we will shortly see,
all $\BSPACE$ oracles can be replaced by $\CSPACE$ machines with the appropriate parameters.

Perhaps more interestingly, the reader is encouraged to think of all results
as a $\CSPACE[O(s),O(c)]$-computable \textit{reduction} of $\CBSPACE[s,c]$ to $\BSPACE[O(s)]$; there is
only a single oracle query, and it outputs the solution to the original function.
We avoid stating it as a reduction only because we require postprocessing to reset
the catalytic tape, which, we note, can be done regardless of the query output as the
oracle halts with no changes to the state of the machine. This reduction cleanly separates manipulating the configuration graph of the catalytic machine from solving the corresponding decision problem on a \say{nice} configuration graph. Previous structural works ~\cite{BuhrmanKouckyLoffSpeelman18,DattaGuptaJainSharmaTewari20,CookLiMertzPyne24} were not able to obtain this separation, and this is why our proof is simpler despite giving stronger results.

\subsection{Derivation of Results}
Then our main results follow directly from existing simulations of randomized and non-deterministic (standard) space, which we now recall:
\begin{theorem}[\cite{BuhrmanCleveKouckyLoffSpeelman14}]\label{thm:NLinCL}
    For all $s := s(n) \geq \log n$,
    $$\NSPACE[O(s)] \subseteq \CSPACE\left[O(s),2^{O(s)}\right].$$
\end{theorem}
\begin{theorem}[\cite{BorodinCookPippenger83,AllenderOgihara96},\cite{BuhrmanCleveKouckyLoffSpeelman14}]\label{thm:PLinCL}
    For all $s := s(n) \geq \log n$,
    $$\PrSPACE[O(s)] \subseteq \CSPACE\left[O(s),2^{O(s)}\right].$$
\end{theorem}
\noindent
Note that the latter result is not stated anywhere but is easily derived.
\begin{theorem}[\cite{Savitch70}]\label{thm:savitch}
    For all $s := s(n) \geq \log n$,
    $$\NSPACE[O(s)] \subseteq \SPACE[O(s^2)].$$
\end{theorem}
\begin{theorem}[\cite{Immerman88,Szelepcs_nyi_1988}]\label{thm:IS}
    For all $s := s(n) \geq \log n$,
    $$\coNSPACE[O(s)] \subseteq \NSPACE[O(s)].$$
\end{theorem}
\begin{theorem}[\cite{Pyne24}]\label{thm:BPLinCL}
    For all $s := s(n) \geq \log n$ and all $\alpha \in [0,1/2]$,
    $$\BPSPACE[O(s)] \subseteq \CSPACE\left[O(s^{1+\alpha}),O(s^{2-\alpha})\right].$$
\end{theorem}
\begin{theorem}[\cite{BorodinCookPippenger83}]\label{thm:BCP}
    For all $s := s(n) \geq \log n$,
    $$\PrSPACE[O(s)] \subseteq \SPACE[O(s^2)].$$
\end{theorem}

From this we can immediately derive all our corollaries using~\Cref{thm:main-tech}.
In particular: \Cref{cor:collapse-cnl} follows from~\Cref{thm:NLinCL};
\Cref{cor:collapse-cpl} follows from~\Cref{thm:PLinCL};
\Cref{cor:cat-savitch} follows from~\Cref{thm:savitch};
\Cref{cor:cat-IS} follows from~\Cref{thm:IS};
\Cref{cor:cat-derand} follows from~\Cref{thm:BPLinCL}; and
\Cref{cor:cat-derand-pl} follows from~\Cref{thm:BCP}.

\subsection{Discussion}
Before going into the proof of \Cref{thm:main-tech}, we
note some corollaries and extensions of our main results.

\paragraph{Catalytic hierarchies.}
All our results can be scaled up to the
\textit{non-deterministic catalytic hierarchy}, defined by classes $\SigmaCSPACE$
and $\PiCSPACE$, as well as the \textit{randomized non-deterministic catalytic
hierarchy}, defined by classes $\MACSPACE$ and $\AMCSPACE$.
While $\SigmaLi=\NL \subseteq \CL$ for all $k$, and
Sdroievski~\cite{Sdroievski24} showed that $\MAL$ is contained
in $\CL$, there are no previously known results connecting the
catalytic non-deterministic hierarchies to $\CL$, even under
the assumption that $\CNL = \coCNL$.

\begin{theorem}~\label{cor:hierarchy}
    $$\CL = \bigcup_{k \in \mathbb{N}} \SigmaCLi = \bigcup_{k \in \mathbb{N}} \MACLi \quad (= \bigcup_{k \in \mathbb{N}} \PiCLi = \bigcup_{k \in \mathbb{N}} \AMCLi).$$
\end{theorem}
\begin{theorem}\label{cor:hierarchy-cspace}
    For all $s := s(n)$, $c := c(n)$ such that $\log n \leq s \leq c \leq 2^s$
    and for all $k \in \mathbb{N}$,
    \begin{align*}
        \SigmaCSPACE[s,c], \PiCSPACE[s,c] &\subseteq \CNSPACE[O(s),O(c)]\\
        \MACSPACE[s,c], \AMCSPACE[s,c] &\subseteq \CNSPACE[O(s),O(c+s^2)].
    \end{align*}
\end{theorem}

Of note, space-bounded hierarchy machines need to be defined carefully,
as too much access to
the various quantifiers at different points of time can result in a sharp
increase in power.
We define these machines as taking quantified variables $y_1 \ldots y_k$,
with the appropriate notions of accepting or rejecting over the choices of
$y_i$, but with the additional restriction that the quantified variables are
written from outermost to innermost on the non-deterministic tape, meaning each $y_j$ can be accessed
in a read-once fashion, and no $y_j$ can be read after $y_{j'}$ for $j' > j$.
For this definition, it is fairly straightforward to show that even the
unbounded-depth hierarchy is only moderately strong:
\begin{theorem}
    For all $k \leq \poly n$,
    \begin{align*}
        \SigmaCLi, \PiCLi &\subseteq \ZPP\\
        \MACLi, \AMCLi &\subseteq \BPP.
    \end{align*}
\end{theorem}

\paragraph{Catalytic and non-catalytic space.}
While results in catalytic space have been more forthcoming than
their classical space counterparts in recent years,
it is unclear whether proving connections between ordinary
space classes is \textit{formally} any harder (or easier) than
proving connections between the corresponding catalytic space classes.

A corollary of our reduction is that
ordinary space and catalytic space now share the same fate with regards to
the power of additional resources:
\begin{corollary}\label{cor:space-cspace}
    Let $\mathcal{B}_1,\mathcal{B}_2 \in \{\perp,\text{\upshape\textsf{N}},\text{\upshape\textsf{coN}},\text{\upshape\textsf{U}},\text{\upshape\textsf{BP},\text{\upshape\textsf{Pr}}}\}$.
    Then
    \begin{align*}
    \mathcal{B}_1\SPACE[O(s)] & \subseteq \mathcal{B}_2\SPACE[O(s)] \quad \text{iff} \\
    \forall c \geq s, \quad \mathcal{B}_1\CSPACE[O(s),O(c)] & \subseteq \mathcal{B}_2\CSPACE[O(s),O(c)]
    \end{align*}
    where $\perp$ indicates no additional resources.
\end{corollary}
\noindent
We make a note about unambiguity here.
The proof of \Cref{thm:main-tech} extends to $\CUSPACE$ as well, but
as this does not prove any new results we did not include it in our statement;
however, it does allow us to get the relevant extension in \Cref{cor:space-cspace}.
This gives the consequence that $\NL=\UL$ holds iff
$\CNSPACE[s,c]\subseteq \CUSPACE[O(s),O(c)]$ for every $s$ and $c$.

\paragraph{Lossy catalytic space.}
We lastly note one use of our result in the context of another catalytic model,
namely \textit{lossy} catalytic space~\cite{GuptaJainSharmaTewari24,FolkertsmaMertzSpeelmanTupker25}.
Folkertsma et al.~\cite{FolkertsmaMertzSpeelmanTupker25} showed that
allowing errors when resetting the catalytic tape of a non-deterministic
or randomized catalytic machine
is equivalent to giving the machine extra free space instead; unfortunately
they could not take the further step of equating these error-free classes
to deterministic ones, as 1) for non-determinism no such connections were known,
and 2) the results of \cite{CookLiMertzPyne24} could no longer
be applied to derandomize after adding this extra space.

Neither of our main theorems, i.e. Theorems \ref{cor:collapse-cnl} and
\ref{cor:collapse-cpl}, are robust enough to draw such a connection.
However, Theorems \ref{cor:cat-savitch}, \ref{cor:cat-derand}, and \ref{cor:cat-derand-pl}
\textit{are} robust
to $s \gg \omega(\log c)$, and so we immediately get the following
results (for definitions and motivation see \cite{FolkertsmaMertzSpeelmanTupker25}):
\begin{corollary} \label{cor:lossy}
    For all $s := s(n)$, $c := c(n)$, $e := e(n)$ such that $\log n \leq s \leq c \leq 2^s$
    and $e \leq \sqrt{c}$, and for all $\alpha \in [0,1/2]$,
    \begin{align*}
    \LCNSPACE[s,c,e] & \subseteq \CSPACE[O((s + e \log c)^2),O(c)] \\
    \LCBPSPACE[s,c,e] & \subseteq \CSPACE[O((s + e \log c)^{1+\alpha}),O(c+(s+e \log c)^{2-\alpha})] \\
    \LCPSPACE[s,c,e] & \subseteq \CSPACE[O((s + e \log c)^2),O(c)]
    \end{align*}
\end{corollary}

\section{Proof of Main Result}

\newcommand{\rot}{{\textsc{Rot}}}
\newcommand{\nextstep}{{\textsc{Next}}}
\newcommand{\stepback}{{\textsc{StepBack}}}
\newcommand{\walk}{{\textsc{Walk}}}
\newcommand{\countstepsback}{{\textsc{CountStepsBack}}}
\newcommand{\countsize}{{\textsc{Size}}}

\newcommand{\walkback}{{\textsc{WalkBack}}}
\newcommand{\confbit}{{\textsc{ConfBit}}}
\newcommand{\canon}{{\textsc{Canon}}}
\newcommand{\countunique}{{\textsc{CountUnique}}}
\newcommand{\findtarget}{{\textsc{FindTarget}}}
\newcommand{\badhashindex}{{\textsc{BadHashIndex}}}
\newcommand{\badhashdescription}{{\textsc{BadHashDescription}}}
\newcommand{\comporcomp}{{\textsc{ComputeOrCompress}}}
\newcommand{\decompress}{{\textsc{Decompress}}}

\newcommand{\acc}{{\textsc{Acc}}}
\newcommand{\rej}{{\textsc{Rej}}}
\newcommand{\tcompr}{{\textsc{ComprTarget}}}
\newcommand{\hcomprDFS}{{\textsc{ComprHashDFS}}}
\newcommand{\hcomprGS}{{\textsc{ComprHashGS}}}

In this section we prove our central technical theorem.
We begin with a discussion of the structure of catalytic machines, followed
by an overview of our approach, and lastly we fill in the details to formally
prove \Cref{thm:main-tech}.

\subsection{Configuration Graphs of Catalytic Machines}
Let $[n] = \{0,1,\cdots,n-1\}$.
For a graph $\cG$, we denote its vertex set by $V(\cG)$.
We use $x\cdot y$ to represent the concatenation of strings $x$ and $y$.

Let $\cM$ be a valid catalytic machine computing $f$.
We will assume without loss of generality that all auxiliary information
about the current configuration of $\cM$, i.e. the state of $\cM$'s internal DFA,
the current positions of tape heads for the input, work, and catalytic tapes
are all automatically recorded in a designated part of the worktape.\footnote{Altogether this additional information
technically requires additional space $\log n + \log s + \log c + O(1)\le 3s$, 
we can handle this by replacing $s$ with $4s$ throughout the proofs,
which we omit for clarity.}
(The contents of the output/oracle tape and its head position is not considered to be a part of a machine configuration.)

\begin{definition}
    Let $\cM$ be a catalytic machine with work space $s$ and catalytic space $c$.
    We denote by $\conf{\pi}{u}$ the configuration of $\cM$ where $\pi \in \{0,1\}^c$ is contained on the catalytic tape and $u \in \{0,1\}^s$ is on its work tape.     
\end{definition}

Consider the execution of $\cM$ on some fixed input $x$ and initial catalytic tape
contents $\tau$. Each configuration of $\cM$ can be uniquely represented by $\conf{\pi}{u}$
for some $\pi \in \{0,1\}^c$ and $u \in \{0,1\}^s$.
Without loss of generality we define
$\start := \conf{\tau}{0^s}$ to be the start configuration and $\accept := \conf{\tau}{1 \cdot 1 \cdot 0^{s-2}}$  
to be the unique accepting halt configuration, and $\reject := \conf{\tau}{1 \cdot 0 \cdot 0^{s-2}}$ to be the unique rejecting halt configuration.

It will often be useful to talk about the configuration graph defined
by such executions.

\begin{definition}[Configuration graphs]\label{definition_config_graphs}
    The \emphdef{configuration graph} $\cGM$ is the directed acyclic graph
    where each node corresponds to a configuration of $\cM$ on input $x$, where there is a directed edge from $\conf{\pi}{u}$ to $\conf{\pi'}{u'}$ iff $\conf{\pi'}{u'}$
    can be reached from $\conf{\pi}{u}$ in one execution step of $\cM$. The out-degree of every vertex in $\cGM$ is at most 2, and there is some fixed constant $d_\cM$ depending only on $\cM$ such that each vertex in $\cGM$ has in-degree at most $d_\cM-2$.
    We call the outgoing edges \emphdef{forward} edges of $\ctv$.
    The remaining edges of $\conf{\tau}{v}$ are \emphdef{backward} edges.
    A halting configuration has no forward edges in $\cGM$.
    
    We say an edge $(v,v')$ is labeled with $b\in \zo$ if it corresponds to a non-deterministic/randomized $b$-choice of the machine. For deterministic transitions we label the edge with both $0$ and $1$. 

    For every catalytic tape $\tau$ let $\cGMt$ be the subgraph of $\cGM$ induced on configurations of $\cGM$ that are reachable from $\start$. Clearly $\cGMt$ has one source node, namely $\start$, and up to two sink nodes,
    namely $\accept$ and $\reject$. 
\end{definition}

\noindent
Our main method of exploring non-deterministic graphs will be to simply set
our non-deterministic sequence to the all-zeroes string:
\begin{definition}[$0$-graph of configuration graphs]
    Given a configuration graph $\cGM$, we let the \emphdef{$0$-graph $\cGMz$} be the undirected graph where only edges with label $0$ are retained, and we forget the direction of each edge. Observe that for every $\tau$ and $v\in \cGMt$,  $\cGMz(v)$ is a tree as each node has at most one forward edge corresponding to the edge labeled 0 in $\cGM$. 
    Given a configuration $v$, we let $\cGMz(v)$ be the connected component of $\cGMz$ containing $v$.
\end{definition}

\noindent
The following fact is immediate:
\begin{fact}
    Let $v, v'$ be such that $v \in \cGMz(v')$. Then $\cGMz(v) = \cGMz(v')$ and
    hence $v' \in \cGMz(v)$.
\end{fact}

\subsection{Proof Overview}
\label{sec:overview}

Given a machine $\cM$, input $x$, and starting catalytic tape $\tau$,
our focus will be on the 0-graphs reachable from our unique halting states
$\accept,\reject$.
These graphs together include all states reachable from the start configuration:
\begin{equation}\label{eq:partition}
    V(\cGMt) \subseteq V(\cGMz(\accept))\cup V(\cGMz(\reject)).
\end{equation}
This follows as for every configuration $v$ that can be reached from running forward from $\start$, it must be the case that running forward from $v$ on the all-$0$s auxiliary input reaches $\accept$ or $\reject$ (as otherwise the machine would not be valid). Moreover, we can deterministically and reversibly explore $\cGMz(\accept)$ and $\cGMz(\reject)$, since both subgraphs are trees and the roots, which are halting configurations, can be identified by examining the worktape contents. If the total size of $\cGMz(\accept) \cup \cGMz(\reject)$ is bounded by $2^{O(s)}$, then we could use our $\BSPACE$ oracle to solve our function on this graph by identifying each node in the union with its index in the exploration of the two graphs.

Unfortunately, if both graphs are large, this exploration may produce a graph which is too large even to write to the query tape. 
To avoid this issue, we adopt an idea of Cook et al.~\cite{CookLiMertzPyne24} to design what they
dubbed a \textit{compress-or-compute} argument.
We observe that the average component size is bounded:
\begin{equation}\label{eq:zdfsSize}
    \E_{\tau}[|V(\cGMz(\accept))\cup V(\cGMz(\reject))|] \le 2^s.
\end{equation}
Similar average-case bounds~\cite{BuhrmanCleveKouckyLoffSpeelman14,DattaGuptaJainSharmaTewari20,CookLiMertzPyne24} have been observed before, but with an important difference -- they all focused on bounding the average size of the forward-reachable graph $\cGMt$. 

The strategy of~\cite{CookLiMertzPyne24} for a deterministic $\CL$ machine worked as follows. We pad the catalytic tape $\tau$ with $s+1$ bits, producing a new tape $(\tau,i)$, and interpret $i\in [2^{s+1}]$ as a counter. They then explore $\cGMz(\start)$. If this exploration reaches a halt configuration within $2^{s+1}= \poly(n)$ steps, we successfully decide the language, and revert $\tau$. Otherwise, let $\cpu$ be the $i$-th configuration encountered on this traversal. Note that $|\pi,u|=|\tau,i|-1$, and hence we can compress our initial catalytic tape by $1$ bit by storing $(\pi,u)$. To revert the tape, we run the machine backwards from $\cpu$ until we encounter the start state (which reverts $\tau$), and count the number of steps to reach it (which is $i$).

However, for a randomized catalytic machine traversing $\cGMz$ from the start state may not explore all vertices required to decide the language, even if $\cGMt$ is small, as it ignores all $1$ transitions. To prove $\CBPL=\CL$, Cook et al.~\cite{CookLiMertzPyne24} built a much more complicated argument based on taking random walks from the start state, where these walks were themselves generated by a reconstructive PRG~\cite{NisanWigderson94,DoronPyneTell24}. This added substantial complexity and limited their findings to the case where \( c = \text{poly}(2^s) \). Furthermore, their results could not accommodate nondeterministic catalytic machines, as missing even a single state in the random walks could make it impossible to decide the language.

We observe that all of this complexity can be completely removed if we instead reversibly traverse from the \emph{halt} states. For a tape $\tau$, we explore $\cGMz(\accept)$ and $\cGMz(\reject)$. If both graphs are of size at most $2^{s+1}$, we can construct the union graph and determine connectivity in $\cGMt$ (or whatever else is required to decide the language) via our oracle. Otherwise, suppose WLOG that $|V(\cGMz(\accept))|\ge 2^{s+1}$. Then we adopt the strategy of~\cite{CookLiMertzPyne24}.  We pad our tape to $(\tau,i)$ and let $\cpu$ be the $i$-th configuration in this traversal. We compress
\[
(\tau,i) \ra (\pi,u,0)
\]
and revert essentially as in~\cite{CookLiMertzPyne24}, by traversing backwards and counting steps until we reach a halting configuration.

An iterative application of this idea either decides the language or frees $c$ bits on the catalytic tape. If the latter occurs, we can brute force over tapes $\tau'$ until we find one for which this graph is small---such a $\tau'$ must exist by~\Cref{eq:zdfsSize}---which we can then provide to the oracle. The only difference between models is that the oracle is solving a different problem on the configuration graph.

\subsection{Formal Proof}
We now formalize our discussion from \Cref{sec:overview}.
Fix a valid catalytic machine $\cM$ using catalytic space $c$ and work space $s$
computing a language $L\in \CBSPACE[s,c]$, and fix an $n$-bit input $x$.
The following fact is immediate:
\begin{fact}
    Each tree in $\cGMz$ contains at most one halting configuration.
\end{fact}
\noindent
It follows, then, that the trees in $\cGMz$ are pairwise disjoint:

\begin{fact}\label{prop:disjointDFS}
    Let $\tau\neq \tau' \in \{0,1\}^c$ be two distinct contents of the catalytic tape of a catalytic machine $\cM$ on an input $x$.
    The vertex sets of $\cGMz(\accept)$, $\cGMz(\acceptp)$, $\cGMz(\reject)$ and $\cGMz(\rejectp)$
    are pairwise disjoint.
\end{fact}
\noindent
An immediate consequence of the above is that the average size of
these components is bounded:
\begin{lemma}\label{prop:sizeexpectation}
    Let $\cM$ be a catalytic machine with work space $s := s(n)$ and
    catalytic space $c := c(n)$, where $\log n \le s \le c \le   2^s$. Then
    $$
        \Exp_{\tau\in \{0,1\}^c}[ |V(\cGMz(\accept))|]\le 2^s \text{ and } \Exp_{\tau\in \{0,1\}^c}[ |V(\cGMz(\reject))|] \le 2^s.
    $$
\end{lemma}
\noindent
Finally, the forward-reachable graph from every state is contained in the trees rooted at the two possible halt states:
\begin{lemma}\label{lem:reachInDFS}
    Let \(\mathcal{M}\) be a valid catalytic machine, and let \(\cpu\) be an arbitrary node in \(\cGMt\). Then, $\cpu \in V(\cGMz(\accept)) \cup V(\cGMz(\reject))$.
\end{lemma}
\begin{proof}
    Since $\cpu$ is reachable by $\cM$ on $x$ from $\start$, there exists some
    non-deterministic sequence $\sigma$ such that $\cM$ reaches configuration $\cpu$.
    Now consider the non-deterministic sequence $\sigma \cdot 0^*$, which takes us
    to $\cpu$ and subsequently uses 0 as its non-deterministic choices. Since we started
    from $\start$, this must eventually either reach $\accept$ or $\reject$,
    and since the latter part of this walk resides entirely inside $\cGMz(\cpu)$,
    either $\cGMz(\accept)$ or $\cGMz(\reject)$ must contain $\cpu$.
\end{proof}

\paragraph{Exploration of Graphs.}
We formally define the notation related to graphs for the purpose of graph exploration.
For an undirected graph $G=(V,E)$ and $v\in V$, we let $G(v)$ be the component of $v$.
Let $d$ be the maximum degree of $G$.

We assume that there is a cyclic ordering of edges at each vertex that define a rotation map $\rot: V \times [d] \rightarrow V \times [d]$,
such that $\rot(v,i)=(u,j)$ if the $i$-th edge of $v$ is the $j$-th edge of $u$
(if $i\ge \deg(v)$ then $\rot(v,i)=(v,i)$). 
If we identify the $i$-th directed edge leaving $v$ by $(v,i)$, then $\rot(v,i)$ flips the direction of the edge.

We will consider walks on $G$ starting from a given edge $(v,i)$ that follow an Eulerian tour of $G(v)$.
For $v \in V$ and $i \in [\deg(v)]$, the next edge of the walk is given by $\nextstep(v,i)=(u,j+1 \mod \deg(v))$ where $\rot(v,i)=(u,j)$.
A step back is taken by $\stepback(u,j)=\rot(u,j-1 \mod \deg(u))$.

We will index edges at each vertex (configuration) $v$ in $\cGMz$ by $[\deg(v)]$ so that the forward edge gets index 0 (if there is a forward edge).
We fix the rotation map $\rot^0$ of $\cGMz$ arbitrarily otherwise.

\paragraph{Defining Catalytic Subroutines.}
We can define a catalytic subroutine $\rot(\cpu,i)$ which uses space $O(s)$ that,
given $\pi\in \{0,1\}^c$ on the catalytic tape, and $u\in \{0,1\}^s$ and
$i \in [d_\cM]$ on the worktape, replaces $(\cpu,i)$ by $(\cpup,j)=\rot^0(\cpu,i)$.
Clearly, $\nextstep(\cpu,i)$ and $\stepback(\cpup,j)$ can be implemented by
catalytic subroutines working in space $O(s)$.

Let $S=2^{O(s)}$ be a (easily computable) function of $n$.
We can also define a catalytic subroutine $\walk(\ctv,i,t)$ which applies the subroutine $\nextstep(\cdot)$ $t$ times on $(\ctv,i)$ where $t \le S$. 
The procedure uses an additional work space of size $O(s)$, and the input $(\ctv,i)$ is replaced by the output $(\cpu,j)$ on the respective tapes. 

We will call the subroutine $\walk(\ctv,i,t)$ for halting configurations $\ctv$ and $i=0$.
For these calls, we can define an inverse subroutine $\countstepsback(\cpu,j)$ which
calculates $\ell \le S$, the number of times we need to apply $\stepback(\cdot)$ on
$(\cpu,j)$ before reaching $(\ctv,0)$ for some halting configuration $\ctv$.
This subroutine uses extra $O(s)$ work space, it replaces $(\cpu,j)$ by $(\ctv,0)$,
and returns the count to a designated area of the work space.
If the count is bigger than $S$ it returns $\infty$.

Combining $\walk(\ctv,i,t)$ with $\countstepsback(\cpu,j)$ we can create a
catalytic subroutine $\confbit(\ctv, b, t)$ which, given a halting configuration $\ctv$, $b\le c+s$ and $t\le S$,
determines the $b$-th bit of the configuration reached by $\walk(\ctv,0,t)$.
For a halting configuration $\ctv$, $\confbit(\ctv, b, t)$ preserves $\ctv$ and $t$
on its tape when it finishes its computation.
Additionally, we define a subroutine $\canon(\ctv,i,t)$ which preserves $\ctv,i,t$ and returns 1, if the edge  $(\cpu,j)$ reached by $\walk(\ctv,i,t)$ has $j=0$, and it returns 0 otherwise.
We think of $t$ as the \textit{canonical} index of the configuration $\cpu$ within $\cGMz(\ctv)$.

Similarly, we can define a catalytic subroutine $\countsize(\ctv)$, which, for a halting configuration $\ctv$, determines the minimum number of steps $t\ge 1$, such that $\walk(\ctv, 0, t)$ returns back to $(\ctv, 0)$. 
If $t$ is at most $S$, it outputs $t$; otherwise, it outputs $\infty$. 
Since $\cGMz(\ctv)$ forms a tree for a halting configuration $\ctv$, the subroutine returns twice the number of edges of $\cGMz(\ctv)$ iff $2\le |V(\cGMz(\ctv))| \le \frac{S}{2} + 1$.  
It returns $1$ or $\infty$ otherwise.
Additionally, the subroutine uses $O(s)$ extra work space.

Note that all the above procedures should ignore any portions of the machine tapes
not directly referenced therein.

\paragraph{The Main Subroutine.}
Our plan is to use the Compress-or-Compute strategy.
Given a starting catalytic tape $\tau$ for a machine $\cM$,
we will either use $\cGMz(\accept)$ and $\cGMz(\reject)$ to construct a small graph
that determines the outcome of the computation of $\cM$ on $x$,
or we will use vertices in $\cGMz(\accept)$ and $\cGMz(\reject)$ to compress
the catalytic tape. We can state the main Compress-or-Compute lemma:

\begin{lemma}\label{lemma:compressorcompute}
   Let $\cM$ be a catalytic machine with work space $s := s(n)$ and
   catalytic space $c := c(n)$, where $\log n \le s \le c \le  2^s$,
   and let $x \in \{0,1\}^n$ be an input for $\cM$ written on the input tape.
   Let $B={2s}$ and $S=2^B$, and
   let $\tau \in \{0,1\}^c$ and $\tar \in \{0,1\}^B$ be given on the catalytic tape.
   There is a catalytic subroutine $\comporcomp(\tau,\tar)$ which takes one of
   the following two actions:
   \begin{enumerate}
       \item{Compute:}  If both $\cGMz(\accept)$ and $\cGMz(\reject)$ are of size at most $S/2+1$, then it outputs a directed graph $G$ and two vertices $r$ and $t$ such that the forward reachable graph from $r$ is isomorphic to $\cGMt$, with $\start$ mapping to $r$ and $\accept$ mapping to $t$.
       \item{Compress:} Otherwise, it replaces $\tau$ by $\pi \in \{0,1\}^c$ and $\tar$ by $(u,j,0^{s- \log d_\cM})$, where $u\in\{0,1\}^s$ and $j\in[d_\cM]$ have the property that $\countstepsback(\cpu,j)$  replaces $\pi$ by $\tau$ and returns $\tar$ as the number of steps.
   \end{enumerate}
   The subroutine returns a bit indicating which action it took,
   and the procedure leaves other portions of the tapes unchanged.
   Furthermore $\comporcomp(\tau,\tar)$ uses additional space $O(s)$.
\end{lemma}

\begin{proof}
Recall that $\cGMt$ is a subgraph of $\cGM$ induced on configurations of $\cGM$ reachable from $\start$, and that
by~\Cref{lem:reachInDFS} we have
\[V(\cGMt) \subseteq V(\cGMz(\accept) \cup V(\cGMz(\reject)).\]

In brief, if both $\cGMz(\accept)$ and $\cGMz(\reject)$ are of size at most $\frac{S}{2}+1$, we can explore them completely using $\walk(\cdot)$, and reconstruct a graph $G$ containing $\cGMt$. 
If  $\cGMz(\accept)$ or $\cGMz(\reject)$ is large we can compress $\tau$.

\medskip\noindent
\textit{Initial check:}
We first check the sizes of $\cGMz(\accept)$ and $\cGMz(\reject)$ using calls to
$\countsize(\accept)$ and $\countsize(\reject)$.
Since $\accept = \conf{\tau}{1 \cdot 1 \cdot 0^{s-2}}$ and
$\reject = \conf{\tau}{1 \cdot 0 \cdot 0^{s-2}}$, both states are easy to prepare given
$\start$, and $\countsize$ can be run in $O(s)$ space.
If either of the sizes exceeds $S/2 + 1$, meaning if either call to $\countsize$ returns $\infty$, we move to the \textit{compress case}; otherwise, we proceed to the \textit{compute case}.

\medskip\noindent
\textit{Compute case:}
If both the graphs have a size of  at most $S/2+1$, we can explore configurations of $\cGMz(\accept)$ and
$\cGMz(\reject)$ using $\confbit(\cdot)$.
We will index the configurations of $G$ by $[S]\times [2]$.
The configuration indexed $(i,b)$ is the configuration reached by $\walk(\accept,0,i)$
if $b=0$ and by $\walk(\reject,0,i)$ otherwise.

For each $(i,b), (j,d) \in [S]\times [2]$, we can check
whether there is an edge from the configuration $(i,b)$ to $(j,d)$ in $\cGM$ by
comparing them bit-by-bit using $\confbit(\cdot)$.
If so and $i$ and $j$ are canonical indexes of their respective configurations (which can be checked by calling $\canon(\cdot)$) we connect them by an edge in $G$.
Hence, we output a graph $G$ on $[S]\times [2]$ where the connectivity between the
canonical indexes of configurations from $\cGMt$ is the same as in $\cGMt$.
By checking each $(i,b) \in [S]\times [2]$, we can locate a canonical copy of a configuration
$\start$ and $\accept$, and output them as $r$ and $t$. 

This computation will use at most $O(s)$ space on the work tape to run $\walk$ and $\confbit$,
and it will preserve $\tau$ and $\tar$ on the catalytic tape.

\medskip\noindent
\textit{Compress case:}
Consider without loss of generality the case where $\countsize(\accept)$ returns $\infty$. 
We prepare $\accept = \conf{\tau}{v}$ where $v=1 \cdot 1 \cdot 0^{s-2}$, and
run $\walk(\accept,i,\tar)$ with $i$ set to 0, treating $\tar$ as a natural number evaluated in base-2, plus one. The result of $\walk$ will be to
replace $\tau$ by some $\pi$, $v$ by some $u\in \{0,1\}^s$, and $i$ by some $j$.
We replace $\tar$ by $(u,j,0^{s-\log d_\cM})$ and end the procedure.

This computation utilizes at most \( O(s) \) workspace, which is all that is needed for the subroutine \(\walk\). Since \(\countsize(\accept)\) returns \(\infty\), it indicates that during the first \( S \) steps prescribed by \(\walk\), we do not return to the edge \((\accept, 0)\). Therefore, given that \(\tar \leq S\) (with \(\tar\) treated as a natural number), calling \(\countstepsback(\cpu, j)\) replaces \(\pi\) with \(\tau\) and returns \(\tar\) as the number of steps taken. Consequently, the output possesses the required properties.
\end{proof}

We now finish the proof of \Cref{thm:main-tech} using the compression and decompression
procedures from above.

\begin{proof}[Proof of \Cref{thm:main-tech}]
Let $\cM$ be our $\CBSPACE[c,s]$ machine and fix an $n$-bit input $x$.
Define $B := 2s$ and $S := 2^B$.

Our goal is to output a directed graph $G$ and two vertices $r$ and $t$ where $t$
is reachable from $r$ in $G$ iff $\cM$ accepts $x$.
The graph $G$ will be obtained by the Compress-or-Compute subroutine of
\Cref{lemma:compressorcompute} which will be run for a suitable choice of $\tau$.
Given such a graph $G$, we can clearly obtain the answer to our function by appealing
to our oracle, as it will be a graph of size at most $2S = 2^{2s+1}$---thus it can be analyzed by a $\BSPACE[O(s)]$ machine---which represents $\cGMt$.

We let $k \ge  2+ 2c/s$, and we think of our catalytic tape as consisting of blocks
    \[
    (\tau,\tar_0,\tar_1,\ldots,\tar_{k-1})
    \]
where $\tau\in \zo^c$ and $\tar_i \in \zo^{B}$. Note that this gives a total catalytic
length of $c + (2+2c/s) \cdot 2s \leq 10c$ as desired.

We iterate over $i\in [k]$ and call $\comporcomp(\tau_i,\tar_i)$, where $\tau_i$
is the first $c$ bits of the catalytic tape at the time when we begin the
$i$-th iteration; thus $\tau_0 := \tau$.
Each call $\comporcomp(\tau_i,\tar_i)$ either outputs the desired graph $G$
or compresses $\tar_i$.
In the former case, we obtain the graph $G$ on which we can run our oracle to obtain
the solution to our original function, at which point we can decompress (see below).
In the latter case, $\tau_i$ is replaced by some $\pi$, which we refer to as $\tau_{i+1}$,
and $\tar_i$ is replaced by some $(u_i,j_i,0^{s- \log d_\cM})$; we then move on
to the $(i+1)$-st iteration.

If none of the calls gives the desired graph, then since we free at least $s/2$ bits of the
catalytic tape during each iteration, we free at least $c+s$ bits of space on
the catalytic tape in total.
We can use this space to iterate over all possible $\tau_k\in \{0,1\}^c$ and set
$\tar_k=1^B$, and see for which one $\comporcomp(\tau_k,\tar_k)$ falls into the compute case.
Whenever it does not do so, i.e. whenever it falls into the compress case,
then it replaces the current $\tau_k$ by some $\pi$ and $\tar_k$ by some
$(u,j,0^{s- \log d_\cM})$; we will revert it back to $\tau_k$ and $\tar_k$ by running
$\countstepsback(\cpu,j)$, which will replace $\pi$ by $\tau_k$ and return $\tar_k$ as
the number of steps.
We increment $\tau_k$ viewed as a binary counter and continue for our new $\tau_k$.

By \Cref{prop:sizeexpectation},  $\Exp_{\tau\in \{0,1\}^c}[ |V(\cGMz(\accept))|]\le 2^s$
and  $\Exp_{\tau\in \{0,1\}^c}[ |V(\cGMz(\reject))|]\le 2^s$.
Thus, for at least half of the possible starting states $\tau$, we have that
$|V(\cGMz(\accept))| \le 4\cdot 2^s$ and  $|V(\cGMz(\accept))| \le 4\cdot 2^s$. 
In particular, there must exist some $\tau\in\{0,1\}^c$ for which both $V(\cGMz(\accept))$
and $V(\cGMz(\reject))$ are smaller than $\frac{S}{2}$, and on this $\tau_k = \tau$
we reach the compute case and output the desired graph $G$.

Recall that once we find a graph $G$ in the compute case, we can appeal to
our oracle to obtain the answer to our function. If we do so via the $\tau_k$
loop above we then erase $(\tau_k,\tar_k)$ on our tape. We are now
left at the state immediately following $\comporcomp(\tau_i,\tar_i)$ for some
$i \in [k]$; our last step is to decompress
each round of $\comporcomp(\tau_i,\tar_i)$, in reverse order, that we
executed until the final call.

To decompress $\pi=\tau_{i+1}$ and $(u_i,j_i,0^{s- \log d_\cM})$,
we call $\countstepsback(\conf{\tau_{i+1}}{u_i},j_i)$ which will replace $\tau_{i+1}$ by
$\tau_i$ and return $\tar_i$ as the number of steps.
Hence we can restore $\tar_i$, and $\tau_i$, and then we move on
to $i-1$. Our final state will once again be the initial catalytic tape
    \[
    (\tau,\tar_0,\tar_1,\ldots,\tar_{k-1})
    \]
at which point we return our saved answer and halt.
\end{proof}


\ifnames
\section*{Acknowledgements} 
We thank Ninad Rajgopal for discussions relating
to the catalytic hierarchies. 
\fi


\DeclareUrlCommand{\Doi}{\urlstyle{sf}}
\renewcommand{\path}[1]{\small\Doi{#1}}
\renewcommand{\url}[1]{\href{#1}{\small\Doi{#1}}}
\bibliographystyle{alphaurl}
\bibliography{bibliography}

\newcommand{\etalchar}[1]{$^{#1}$}
\begin{thebibliography}{BCK{\etalchar{+}}14}

\bibitem[AO94]{AllenderOgihara96}
Eric Allender and Mitsunori Ogihara.
\newblock Relationships among pl, \#l, and the determinant.
\newblock {\em Proceedings of IEEE 9th Annual Conference on Structure in Complexity Theory}, pages 267--278, 1994.

\bibitem[BCK{\etalchar{+}}14]{BuhrmanCleveKouckyLoffSpeelman14}
Harry Buhrman, Richard Cleve, Michal Kouck{\'{y}}, Bruno Loff, and Florian Speelman.
\newblock Computing with a full memory: catalytic space.
\newblock In {\em ACM Symposium on Theory of Computing (STOC)}, pages 857--866, 2014.
\newblock \href {https://doi.org/10.1145/2591796.2591874} {\path{doi:10.1145/2591796.2591874}}.

\bibitem[BCP83]{BorodinCookPippenger83}
Allan Borodin, Stephen~A. Cook, and Nicholas Pippenger.
\newblock Parallel computation for well-endowed rings and space-bounded probabilistic machines.
\newblock {\em Inf. Control.}, 58(1-3):113--136, 1983.
\newblock \href {https://doi.org/10.1016/S0019-9958(83)80060-6} {\path{doi:10.1016/S0019-9958(83)80060-6}}.

\bibitem[BDS22]{BisoyiDineshSarma22}
Sagar Bisoyi, Krishnamoorthy Dinesh, and Jayalal Sarma.
\newblock On pure space vs catalytic space.
\newblock {\em Theoretical Computer Science (TCS)}, 921:112--126, 2022.
\newblock \href {https://doi.org/10.1016/J.TCS.2022.04.005} {\path{doi:10.1016/J.TCS.2022.04.005}}.

\bibitem[BKLS18]{BuhrmanKouckyLoffSpeelman18}
Harry Buhrman, Michal Kouck{\'{y}}, Bruno Loff, and Florian Speelman.
\newblock Catalytic space: Non-determinism and hierarchy.
\newblock {\em Theory of Computing Systems (TOCS)}, 62(1):116--135, 2018.
\newblock \href {https://doi.org/10.1007/S00224-017-9784-7} {\path{doi:10.1007/S00224-017-9784-7}}.

\bibitem[CLMP25]{CookLiMertzPyne24}
James Cook, Jiatu Li, Ian Mertz, and Edward Pyne.
\newblock The structure of catalytic space: Capturing randomness and time via compression.
\newblock In {\em ACM Symposium on Theory of Computing (STOC)}, 2025.

\bibitem[CM21]{CookMertz21}
James Cook and Ian Mertz.
\newblock Encodings and the tree evaluation problem.
\newblock {\em Electronic Colloquium on Computational Complexity (ECCC)}, {TR21-054}, 2021.
\newblock URL: \url{https://eccc.weizmann.ac.il/report/2021/054}.

\bibitem[CM22]{CookMertz22}
James Cook and Ian Mertz.
\newblock Trading time and space in catalytic branching programs.
\newblock In {\em IEEE Conference on Computational Complexity (CCC)}, volume 234 of {\em Leibniz International Proc. in Informatics (LIPIcs)}, pages 8:1--8:21, 2022.
\newblock \href {https://doi.org/10.4230/LIPIcs.CCC.2022.8} {\path{doi:10.4230/LIPIcs.CCC.2022.8}}.

\bibitem[CM24]{CookMertz24}
James Cook and Ian Mertz.
\newblock Tree evaluation is in space {O}(log n {\(\cdot\)} log log n).
\newblock In {\em ACM Symposium on Theory of Computing (STOC)}, pages 1268--1278. {ACM}, 2024.
\newblock \href {https://doi.org/10.1145/3618260.3649664} {\path{doi:10.1145/3618260.3649664}}.

\bibitem[DGJ{\etalchar{+}}20]{DattaGuptaJainSharmaTewari20}
Samir Datta, Chetan Gupta, Rahul Jain, Vimal~Raj Sharma, and Raghunath Tewari.
\newblock Randomized and symmetric catalytic computation.
\newblock In {\em {CSR}}, volume 12159 of {\em Lecture Notes in Computer Science (LNCS)}, pages 211--223, 2020.
\newblock \href {https://doi.org/10.1007/978-3-030-50026-9\_15} {\path{doi:10.1007/978-3-030-50026-9\_15}}.

\bibitem[DPT24]{DoronPyneTell24}
Dean Doron, Edward Pyne, and Roei Tell.
\newblock Opening up the distinguisher: {A} hardness to randomness approach for {BPL} = {L} that uses properties of {BPL}.
\newblock In {\em ACM Symposium on Theory of Computing (STOC)}, pages 2039--2049, 2024.

\bibitem[FMST25]{FolkertsmaMertzSpeelmanTupker25}
Marten Folkertsma, Ian Mertz, Florian Speelman, and Quinten Tupker.
\newblock Fully characterizing lossy catalytic computation.
\newblock In {\em Innovations in Theoretical Computer Science Conference (ITCS)}, volume 325 of {\em LIPIcs}, pages 50:1--50:13, 2025.

\bibitem[GJST19]{GuptaJainSharmaTewari19}
Chetan Gupta, Rahul Jain, Vimal~Raj Sharma, and Raghunath Tewari.
\newblock Unambiguous catalytic computation.
\newblock In {\em Conference on Foundations of Software Technology and Theoretical Computer Science (FSTTCS)}, volume 150 of {\em Leibniz International Proc. in Informatics (LIPIcs)}, pages 16:1--16:13, 2019.
\newblock \href {https://doi.org/10.4230/LIPIcs.FSTTCS.2019.16} {\path{doi:10.4230/LIPIcs.FSTTCS.2019.16}}.

\bibitem[GJST24]{GuptaJainSharmaTewari24}
Chetan Gupta, Rahul Jain, Vimal~Raj Sharma, and Raghunath Tewari.
\newblock Lossy catalytic computation.
\newblock {\em Computing Research Repository (CoRR)}, abs/2408.14670, 2024.

\bibitem[GKM15]{GirardKouckyMckenzie15}
Vincent Girard, Michal Kouck{\'{y}}, and Pierre McKenzie.
\newblock Nonuniform catalytic space and the direct sum for space.
\newblock {\em Electronic Colloquium on Computational Complexity (ECCC)}, {TR15-138}, 2015.

\bibitem[Hoz21]{Hoza21}
William~M. Hoza.
\newblock Better pseudodistributions and derandomization for space-bounded computation.
\newblock In {\em Proceedings of the 25th International Conference on Randomization and Computation (RANDOM)}, pages 28:1--28:23, 2021.

\bibitem[Imm88]{Immerman88}
Neil Immerman.
\newblock Nondeterministic space is closed under complementation.
\newblock {\em SIAM Journal on Computing (SICOMP)}, 17(5):935--938, 1988.

\bibitem[Kou16]{Koucky16}
Michal Kouck{\'{y}}.
\newblock Catalytic computation.
\newblock {\em Bulletin of the EATCS (B.EATCS)}, 118, 2016.

\bibitem[Mer23]{Mertz23}
Ian Mertz.
\newblock Reusing space: Techniques and open problems.
\newblock {\em Bulletin of the EATCS (B.EATCS)}, 141:57--106, 2023.

\bibitem[MS24]{Sdroievski24}
Nicollas Mocelin~Sdroievski.
\newblock {\em Derandomization vs. Lower Bounds for Arthur-Merlin Protocols}.
\newblock PhD thesis, University of Wisconsin–Madison, Madison, WI, 2024.

\bibitem[Nis92]{Nisan92}
Noam Nisan.
\newblock Pseudorandom generators for space-bounded computation.
\newblock {\em Combinatorica}, 12(4):449--461, 1992.

\bibitem[NW94]{NisanWigderson94}
Noam Nisan and Avi Wigderson.
\newblock Hardness vs.\ randomness.
\newblock {\em Journal of Computer and System Sciences (J.CSS)}, 49(2):149--167, 1994.

\bibitem[Pot17]{Potechin17}
Aaron Potechin.
\newblock A note on amortized branching program complexity.
\newblock In {\em IEEE Conference on Computational Complexity (CCC)}, volume~79 of {\em Leibniz International Proc. in Informatics (LIPIcs)}, pages 4:1--4:12, 2017.
\newblock \href {https://doi.org/10.4230/LIPIcs.CCC.2017.4} {\path{doi:10.4230/LIPIcs.CCC.2017.4}}.

\bibitem[Pyn24]{Pyne24}
Edward Pyne.
\newblock Derandomizing logspace with a small shared hard drive.
\newblock In {\em IEEE Conference on Computational Complexity (CCC)}, volume 300 of {\em LIPIcs}, pages 4:1--4:20, 2024.

\bibitem[RA00]{ReinhardtAllender00}
Klaus Reinhardt and Eric Allender.
\newblock Making nondeterminism unambiguous.
\newblock {\em SIAM Journal on Computing (SICOMP)}, 29(4):1118--1131, 2000.

\bibitem[RZ21]{RobereZuiddam21}
Robert Robere and Jeroen Zuiddam.
\newblock Amortized circuit complexity, formal complexity measures, and catalytic algorithms.
\newblock In {\em IEEE Symposium on Foundations of Computer Science (FOCS)}, pages 759--769. {IEEE}, 2021.
\newblock \href {https://doi.org/10.1109/FOCS52979.2021.00079} {\path{doi:10.1109/FOCS52979.2021.00079}}.

\bibitem[Sav70]{Savitch70}
Walter~J. Savitch.
\newblock Relationships between nondeterministic and deterministic tape complexities.
\newblock {\em Journal of Computer and System Sciences (J.CSS)}, 4(2):177--192, 1970.
\newblock \href {https://doi.org/10.1016/S0022-0000(70)80006-X} {\path{doi:10.1016/S0022-0000(70)80006-X}}.

\bibitem[SZ99]{SaksZhou99}
Michael~E. Saks and Shiyu Zhou.
\newblock {$\bm{\mathsf{BP_{H}SPACE}}[S] \subseteq \bm{\mathsf{DSPACE}}[S^{3/2}]$}.
\newblock {\em JCSS}, 58(2):376--403, 1999.

\bibitem[Sze88]{Szelepcs_nyi_1988}
R\'{o}bert Szelepcs\'{e}nyi.
\newblock The method of forced enumeration for nondeterministic automata.
\newblock {\em Acta Informatica}, 26(3):279–284, 1988.
\newblock URL: \url{http://dx.doi.org/10.1007/BF00299636}, \href {https://doi.org/10.1007/bf00299636} {\path{doi:10.1007/bf00299636}}.

\bibitem[Wil25]{Williams25}
Ryan Williams.
\newblock Simulating time in square-root space.
\newblock {\em Electron. Colloquium Comput. Complex.}, {TR25-017}, 2025.
\newblock URL: \url{https://eccc.weizmann.ac.il/report/2025/017/}.

\end{thebibliography}



\end{document}